%% file: main.tex
\documentclass[submission,copyright,creativecommons]{eptcs}
\providecommand{\event}{SCAV 2018} 
\usepackage{underscore}           

\usepackage{amsmath}
\usepackage{amsfonts}
\usepackage{amssymb}
\usepackage{amsthm}
\usepackage{color}
\usepackage{tikz, pgfplots}
\usepackage{hyperref}
\usepackage{cleveref}
\usetikzlibrary{positioning,arrows,patterns}

\newtheorem{problem}{Problem}
\newtheorem{assumption}{Assumption}

\newtheorem{lemma}{Lemma}
\newtheorem{definition}{Definition}
\newtheorem{example}{Example}
\newtheorem{proposition}{Proposition}

\newcommand{\PP}{\mathcal{P}}
\newcommand{\sig}{[\cdot]}

\newcommand{\RR}{\mathbb{R}}
\newcommand{\Rnn}{{\RR_{\geq 0}}}
\newcommand{\ZZ}{\mathbb{Z}}
\newcommand{\Znn}{{\ZZ}_{\geq 0}}
\newcommand{\NN}{\mathbb{N}}
\newcommand{\nin}{\not\in}
\newcommand{\proj}{{\pi}}
\newcommand{\pow}[1]{{2^{#1}}}
\newcommand{\aset}{\mathcal{A}}
\newcommand{\indep}{\mathsf{IND}}
\newcommand{\maps}{\rightarrow}

\newcommand{\disc}{{\eta}} 

\newcommand{\state}{\mathcal{X}}
\newcommand{\control}{\mathcal{U}}
\newcommand{\inv}{\mathcal{S}} 

\newcommand{\ctrl}{C} 
\newcommand{\block}{\mathcal{B}}
\newcommand{\inter}{I}
\newcommand{\safe}{\mathcal{K}} 

\newcommand{\scp}{\mathsf{SPRE}} 

\newcommand{\cpi}{\mathsf{IPRE}} 
\newcommand{\sipre}{\mathsf{SIPRE}} 
\newcommand{\spre}{\mathsf{SPRE}}
\newcommand{\pre}{\mathsf{PRE}}

\pgfdeclarepatternformonly{soft crosshatch}{\pgfqpoint{-1pt}{-1pt}}{\pgfqpoint{4pt}{4pt}}{\pgfqpoint{3pt}{3pt}}%
{
  \pgfsetstrokeopacity{0.3}
  \pgfsetlinewidth{0.4pt}
  \pgfpathmoveto{\pgfqpoint{3.1pt}{0pt}}
  \pgfpathlineto{\pgfqpoint{0pt}{3.1pt}}
  \pgfpathmoveto{\pgfqpoint{0pt}{0pt}}
  \pgfpathlineto{\pgfqpoint{3.1pt}{3.1pt}}
  \pgfusepath{stroke}
}

    \pgfplotsset{
    standard/.style={
        axis x line=middle,
        axis y line=middle,
        enlarge x limits=0.15,
        enlarge y limits=0.15,
        every axis x label/.style={at={(current axis.right of origin)},anchor=north west},
        every axis y label/.style={at={(current axis.above origin)},anchor=north east},
        every axis plot post/.style={mark options={fill=white}}
        }
    }

\title{Automatic Generation of Communication Requirements for Enforcing Multi-Agent Safety\thanks{This work was supported in part by NSF grant CNS-1545116, co-funded by the DOT.}}
\author{Eric S. Kim \quad Murat Arcak \quad Sanjit A Seshia \institute{Department of Electrical Engineering and Computer Sciences\\UC Berkeley\\
Berkeley, CA}\email{\{eskim, arcak, sseshia\}@eecs.berkeley.edu} \and BaekGyu Kim \quad Shinichi Shiraishi 
\institute{Toyota InfoTechnology Center, U.S.A.\\ ~\\ 
Mountain View, CA} \email{\{bkim, sshiraishi\}@us.toyota-itc.com }}

\def\titlerunning{Automatic Generation of Communication Requirements for Enforcing Multi-Agent Safety}
\def\authorrunning{E. Kim, M. Arcak, S. Seshia, B. Kim \& S. Shiraishi}

\begin{document}

\maketitle
\begin{abstract}
Distributed controllers are often necessary for a multi-agent system to satisfy safety properties such as collision avoidance. Communication and coordination are key requirements in the implementation of a distributed control protocol, but maintaining an all-to-all communication topology is unreasonable and not always necessary.
Given a safety objective and a controller implementation, we consider the problem of identifying \textit{when} agents need to communicate with one another and coordinate their actions to satisfy the safety constraint.
We define a coordination-free controllable predecessor operator that is used to derive a subset of the state space that allows agents to act independently, without consulting other agents to double check that the action is safe. Applications are shown for identifying an upper bound on connection delays and a self-triggered coordination scheme. Examples are provided which showcase the potential for designers to visually interpret a system's ability to tolerate delays when initializing a network connection.
\end{abstract} 

\input{intro}

\input{formulation}

\input{coordconstrained}

\input{intermittent}

\input{examples}

\input{conclusion}

\nocite{*}
\bibliographystyle{eptcs}
\bibliography{refs} 

\end{document}

%% file: intro.tex
\section{Introduction}

Interaction amongst agents can come in various forms such as coupled dynamics, coupling constraints, or a joint optimization objective.
A common facet of multi-agent systems is the use of a distributed control architecture, where each agent has authority over different sets of actuators, and an accompanying communication network for agents to coordinate their actions. 
Communication and collective decision making facilitate complex interactions amongst agents and enable them to reliably achieve collective behaviors that would otherwise be difficult to accomplish without some coordination protocol.

In this paper, we consider the problem of satisfying a safety objective with a controller that is distributed over multiple agents. We say that these agents are coordinating within a given time step if they communicate and collectively agree upon actions to execute.
As a motivating example, consider two fully autonomous vehicles equipped with vehicle-to-vehicle (V2V) communication and tasked with avoiding a collision. 
At one extreme are scenarios where no communication is necessary due to a sufficiently large distance between the vehicles, while at the other extreme are near miss scenarios where collisions are only avoided through precise timing, actuation, or luck. 
Preemptive cooperation enabled by V2V communication is designed to help the vehicles avoid these danger scenarios and for vehicles to negotiate collision-free trajectories. 

How can one distinguish between these extremes and determine when multi-system coordination is and is not necessary to maintain a safety objective?
We present a method that takes a closed loop control system and a safety requirement, then identifies a subset of the state space that is robustly safe against temporary communication losses. This subset naturally shrinks with time as the duration of the communication loss increases.
At its core, our method iterates an appropriate operator which propagates a coordination-free region and resembles  fixed point algorithms in the literature on symbolic system verification. This operator is defined such that it incorporates information about the system dynamics and the controller architecture. 
These results are first used to consider a scenario when multiple agents want to cooperate, but can only do so after some delay. We then develop a self-triggered coordination scheme where agents can preemptively schedule when they would like to communicate, while still maintaining safety guarantees.

This paper tackles a new problem that has not, to the best of our knowledge, been addressed within the control theory literature and is motivated by applications to autonomous vehicle safety. 
Compared to other work, we do not assume a decomposition of the state space as in \cite{dallal16}\cite{chen2016general} nor is the objective assumed to be decomposable \cite{chen2016general}. 
Instead we only consider a decomposition of the input space and can thus accommodate instances when there are complex coupling dynamics that are best handled monolithically. 
This work leverages compositional tools and techniques developed for formal controller synthesis. These may involve constructing abstractions compositionally \cite{rungger2016compositional}, decomposing the controller synthesis procedure \cite{kim2015compositional}\cite{meyer2017compositional}, or decomposing the controller itself \cite{sadraddini2017formal}.
Assume-guarantee reasoning has also been used for compositional synthesis with multiple agents by abstracting out internal information that is irrelevant to reason about system interactions \cite{nuzzo2015platform}. 
Our self-triggering communication scheme may be compared to similar schemes in the self-triggered control literature \cite{heemels2012introduction}, where often the objective is to minimize the energy expended by sensors and actuators subjected to a stability constraint \cite{brunner2015communication}\cite{GOMMANS201559}. Our work instead seeks to minimize the communication overhead incurred as multiple agents negotiate safe actions.

%% file: formulation.tex
\section{Formulation}

\subsection{Notation}

Given two sets $\aset$ and $\mathcal{B}$, let $|\aset|$, $\pow{\aset}$, and $\aset \times \mathcal{B}$ respectively represent $\aset$'s cardinality, $\aset$'s power set (set of all subsets), and the Cartesian product between $\aset$ and $\mathcal{B}$. 
Let $\RR$, $\ZZ$ represent the real and integer numbers respectively, while $\Rnn$ and $\Znn = \NN$ are their non-negative counterparts.
With an appropriate universal set $\Omega$, $\aset$'s complement $\aset^C$ is defined as $\Omega \setminus \aset$.
Given a Cartesian product of $M$ sets $\prod_{i=1}^M \aset_i$ and a subset $L \subseteq \prod_{i=1}^M \aset_i$, the projection operation  $\proj_{\aset_j}: \prod_{i=1}^M \aset_i \maps \aset_j$ retains the coordinates associated with $\aset_j$ and is defined as: 
\begin{align}
\proj_{\aset_j}(L) &= \{a_j \in \aset_j: \exists  (a_1, \ldots, a_{j-1}, a_{j+1}, \ldots, a_M) \text{ such that } (a_1,\ldots, a_M) \in L \}. 
\end{align} 

\subsection{Signals and Systems}
An interval $[a,b]$ where $a,b \in \ZZ$ includes both end points. Let $[a,b) = [a,b-1]$ and $[a] = [a,a]$.
Given a space $\PP$, the space of trajectories evolving in $\PP$ is $\PP\sig$. 
A trajectory $p\sig$ over time interval $\inter$ is a map $p\sig: \inter \maps \PP$. 
Let $\state$ and $\control$ represent a system's state and input spaces respectively. Sets $\state\sig$ and $\control\sig$ are referred to as state and input trajectory sets.
This paper deals with systems where the input space $\control$  consists of $N$ components so that $\control = \prod_{i=1}^N \control_i$ \footnote{Some $\control_i$ may be multi-dimensional so $N$ is not necessarily the dimension of $\control$.}. 
Each of these $N$ components is thought of as an individual agent.
The system's discrete-time dynamics are given by a relation $f \subseteq \state \times \control \times \state$, which can also be viewed as a set-valued function $f: \state \times \control \maps \pow{\state}$. Let $\control(x) = \{u \in \control: f(x,u) \neq \emptyset\} $ denote the set of non-blocking control inputs at $x$. 

A memoryless controller for system $f$ is a relation $\ctrl \subseteq \state \times \control$. 
The set of states $\block = \{x \in \state: (x,u) \nin \ctrl \text{ for all } u \in \control \}$ is the set of blocking states under controller $\ctrl$. 
A controller may also be viewed as a function $\ctrl: \state \maps \pow{\control}$ that maps states to sets of admissible inputs (states with no corresponding control input map to an empty set).
A controller $\ctrl$ and system $f$ can be interconnected into a closed loop system denoted as $f \circ \ctrl: \state \maps \pow{\state}$ \footnote{This notation was inspired by $\circ$'s usage as a function composition operator.  However, it is not a composition in the strictest sense where $f(g(x)) = (f \circ g) (x)$. }. The next state $x[k+1]$ satisfies $x[k+1] \in f \circ \ctrl (x[k]))$ if and only if there exists a $u[k] \in \ctrl(x[k])$ such that $x[k+1] \in f(x[k],u[k])$. All sequences $x\sig$ that satisfy the aforementioned condition and $x[0] \in \mathcal{L}$ are said to be generated by the closed loop system $f \circ \ctrl$ with initial state set $\mathcal{L} \subseteq \state$.

\subsection{Control for Safety}
Safety is a common requirement for cyber-physical systems. We encapsulate this notion of safety as a region of the state space $\inv \subseteq \state$ that should never be exited. 
For a vehicle, set $\inv$ could represent a collision-free zone and a speed limit, while for a medical device $\inv$ could represent safe blood sugar levels. 

\begin{definition}
Let $\inv \subseteq \state$ be a set of safe states. A control policy $\ctrl: \state \maps \pow{\control}$ and initial set $\mathcal{L}\subseteq \inv$ is said to satisfy safety constraint $\inv$ if all trajectories generated by a closed loop system $f \circ \ctrl$ with any initial state $x[0] \in \mathcal{L}$ never exit $\inv$.
\end{definition}

At each state $x$, there is a set of admissible control inputs $\ctrl(x) \subseteq \control$.
A controller is deterministic if $|\ctrl(x)| = 1$ only permits one action for all $x \in \state$. Although determinism simplifies analysis of a closed loop system, deterministic controllers may be too restrictive if the system needs to satisfy additional requirements on top of safety.
For instance if two vehicles want to avoid a collision, then a safe controller can simply enforce that both vehicles have zero velocity but this prevents vehicles from reaching a desired location. 

\subsection{Loss of Safety Guarantees with a Distributed Controller}

More permissive controllers can act as supervisors that restrict control actions only enough to ensure safety. 
They are useful because they can be combined with other controllers that seek to achieve other objectives such as reaching a region.
When a  distributed controller is deployed on multiple systems without an underlying communication scheme, the non-determinism contained in permissive controllers can lead to safety violations.

If $\control = \prod_{i=1}^N \control_i$ is decomposed into $N$ inputs that are each under control from a different agent, then each must concurrently select a single input $u_i$ such that
\begin{align}
(u_1, \ldots, u_N) \in \ctrl(x). \label{eqnsatcontrol}
\end{align}
It is this step where multiple agents  concurrently select an input that leads to coordination hazards.
Whenever $|\ctrl(x)| > 1$ then assuring that (\ref{eqnsatcontrol}) holds is not always possible without explicit coordination and communication with other agents.

\begin{figure}
\centering
\includegraphics[width = .6\columnwidth]{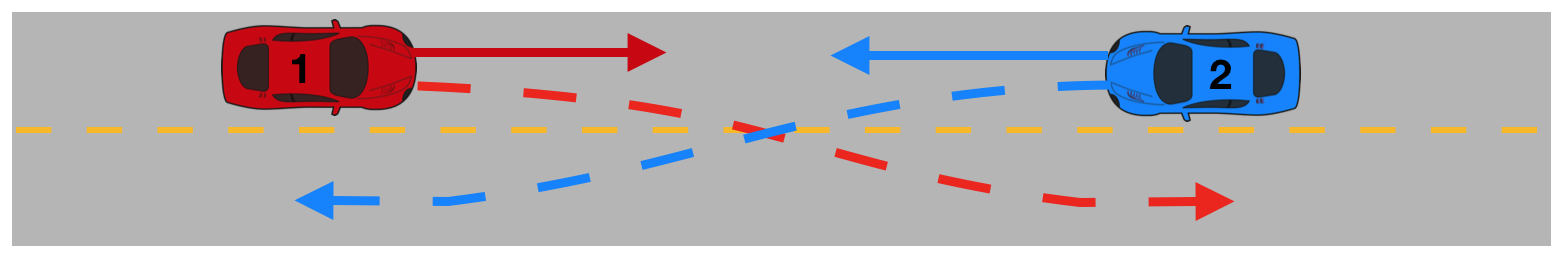}
\begin{small}
\begin{tabular}{|c|c|c|}
\hline & Right Vehicle  &  Right Vehicle \\
& Change & Stay\\
\hline Left Vehicle Change & Collision & No Collision \\
\hline Left Vehicle Stay & No Collision & Collision\\ \hline 
\end{tabular}
\end{small}
\caption{Motivating Example} \label{figcrashref}
\end{figure}

\begin{example}[Illustrative Example]
Consider a scenario depicted in \Cref{figcrashref} where two vehicles are facing one another and a collision is imminent. Both vehicles can choose between staying in their lane or switching to the other lane and a collision is avoided only when one vehicle switches. Clearly it is possible for a collision to be avoided as long as the two vehicles are able to communicate and negotiate which one changes lanes. On the other hand suppose that these vehicles are not equipped with V2V communications. If a collision does occur it is not possible to assign fault to solely one vehicle because from both vehicles' points of view its action was safe as long as the other vehicle responded with the appropriate action. Instead one can only attribute the fault to both agents' failure to negotiate.
\end{example}

To formalize the notion of coordination, we first define a minimal independent controller $\indep_\ctrl$ associated with $\ctrl$. The set of possible controller actions at $x$ is $\indep_\ctrl(x)$ and depicted in \Cref{figindepcontroller}.
\begin{align}
\indep_\ctrl(x) := \prod_{i=1}^N \proj_{\control_i} \ctrl(x).  \label{defindependentC}
\end{align}

\begin{figure}
\centering
\begin{tikzpicture}
\draw[->, thick] (0,0) -- (3.5,0);
\draw[->, thick] (0,0) -- (0,3.5);Sparsity-Aware Finite Abstraction
Gruber, Felix	Tech. Univ. of Munich
Kim, Eric S.	Univ. of California, Berkeley
Arcak, Murat	Univ. of California, Berkeley

\node (u) at (1.5, 3.5) {$\control$};
\node (u1) at (3.5, .25) {$\control_1$};
\node (u2) at (.3, 3.5) {$\control_2$};

\tikzstyle{goodreg} = [thick, draw = black,pattern=soft crosshatch];
\node[circle, goodreg, minimum size=1cm] at (1.2,2.1) {}; 
\node[circle, goodreg, minimum size=.5cm] at (2.2,.8) {};

\tikzstyle{projections} = [line width=2.6pt];
\draw[projections] (.7, 0) -- (1.7,0); 
\draw[projections] (1.95,0) -- (2.45,0);
\node (p1) at (1.5, -.6) {$\proj_{\control_1} (C(x))$};
\draw[->] (p1) -- (1.2,-.1);
\draw[->] (p1) -- (2.2,-.1);

\draw[projections] (0,1.6) -- (0,2.6); 
\draw[projections] (0,0.55) -- (0,1.05);
\node (p2) at (-.8, 1.5) {$\proj_{\control_2} (C(x))$};
\draw[->] (p2) -- (-.1,2.1);
\draw[->] (p2) -- (-.1,.8);

\tikzstyle{indep} = [fill = black, fill opacity = .2];
\draw[indep] (.7,.55) rectangle (1.7, 1.05); 
\draw[indep] (.7,1.6) rectangle (1.7, 2.6); 
\draw[indep] (1.95,1.6) rectangle (2.45, 2.6); 
\draw[indep] (1.95,.55) rectangle (2.45, 1.05); 
\end{tikzpicture}
\caption{For some fixed $x \in \state$, the original safe control set $\ctrl(x)$ (patterned region) is projected onto the axes and yields $\proj_{\control_1}(C(x))$ and $\proj_{\control_2}(C(x))$ (thick lines). Combining the projections gives the coordination-free counterpart $\indep_\ctrl (x)$ (darker regions) defined in \Cref{seccoordregion}.} \label{figindepcontroller}
\end{figure}
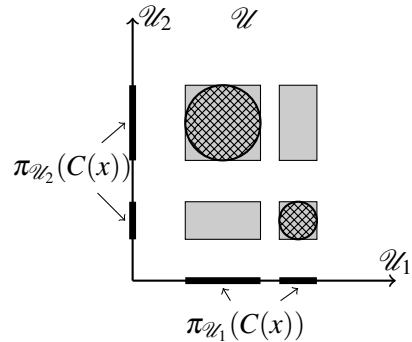

The projection $\proj_{\control_i} C(x)$ of this controller onto each agent $i$'s individual component $\control_i$ yields the set of all control inputs permitted at state $x$ without any information about how other agents behave. 
Any input $u_i \nin \proj_{\control_i} C(x)$ indicates that agent $i$ is either reckless or malicious. 
If all agents pick a $u_i \in  \proj_{\control_i} \ctrl(x)$ then they have all reasonably attempted to satisfy the safety condition by selecting a point $(u_1, \ldots, u_N) \in \indep_\ctrl(x)$, but the joint condition $(u_1, \ldots, u_N) \in \ctrl(x)$ is not necessarily satisfied because $\ctrl(x) \subseteq \indep_\ctrl(x)$.
The independent controller $\indep_\ctrl$ may also be viewed as the set of possible control actions that are reasonable in the undesirable situation where each agent believes itself to be the leader and relies on the other agents to be followers that respond to the leader's choice. 
The set $\indep_\ctrl(x) \subseteq \control $ is the minimal independent set that contains $\ctrl(x)$.

Throughout the rest of this paper, we analyze properties of the new closed loop system $f \circ \indep_\ctrl$, which is derived from $f \circ \ctrl$ but exhibits additional behaviors due to the absence of coordination.

Note that the set of trajectories that are exhibited under $f \circ \ctrl$ is a subset of those exhibited under $f \circ \indep_\ctrl$. Thus, even though the original system $f \circ \ctrl$ may be safe, $f \circ \indep_\ctrl$ may exhibit unsafe trajectories. 

\begin{problem}
Given a set of dynamics $f$, a distributed controller $\indep_\ctrl$, a safe region $\inv$, and coordination-free interval $I = [a,b)$ identify a subset of the state space $\mathcal{L}$ such that all behaviors of $f \circ \indep_\ctrl$ with initial state $x[a] \in \mathcal{L}$ remain in $\inv$ within the interval $I$.
\end{problem}

\subsection{Remarks on Coordination with Mesh Networks}

V2V technology also enables the creation of ad hoc vehicular mesh networks which enables applications in cooperative cruise control, vehicular platoons, and congestion mitigation. 
Suppose each agent is represented by a vertex in an undirected graph and two agents with a V2V have their corresponding vertices connected by an edge. Such a graph can be grouped into equivalence classes corresponding to its connected components. We assume that agents in the same class can communicate instantly even if they are separated by more than one edge. 

\begin{assumption} \label{assumEquivClass}
Each agent in an equivalence class can coordinate with all other agents in that class within each time step $k$.
\end{assumption} 
In practice, \Cref{assumEquivClass} is a requirement that the time scale over which messages is passed in the network are effectively instantaneous relative to the time scale of the physical dynamics. 
The independence definition of \Cref{defindependentC} was stated under the assumption that each $\control_i$ corresponded to one agent and that no agents cooperate. If agent cooperation occurs over a mesh network with $P$ connected components, then the independence condition corresponds to the connected components of the graph. For each of $l = 1,\ldots, P$ equivalence classes, let $\hat{\control}_l$ be the Cartesian product of the coordinates $\control_i$ that belong to that class. 
\begin{align}
\indep_\ctrl(x) := \prod_{l=1}^P \proj_{\hat{\control}_l} \ctrl(x).  \label{defmeshindependentC}
\end{align}
This formulation allows for a platoon to be treated as a single agent instead of a collection of vehicles. For notational simplicity, we simply assume that the decomposition into equivalence classes is given and use \Cref{defindependentC} throughout the rest of this paper.

%% file: coordconstrained.tex
\section{Coordination-Free Operator} \label{seccoordregion}

Given some controller $\ctrl \subseteq \state \times \control$, we use the associated minimally restrictive independent controller from \Cref{defindependentC} as a formal characterization of all the possible actions with a distributed implementation of $\ctrl$ in the absence of coordination.

The set of predecessor states which enforce membership within a region $Z \subseteq \state$ without coordination is computed with the operator
\begin{align}
\cpi (Z) =& \left\{ x: x \in \proj_\state(\indep_\ctrl) \right\} \cap \left\{x: \emptyset \neq f(x,u) \subseteq Z \text{ for all } u \in \indep_\ctrl (x) \right\}. \label{eqnuncoordsafe} 
\end{align} 
The first set ensures that there is always a valid input because $\proj_\state(\indep_\ctrl)$ is a state domain over which the controller produces admissible inputs. The second set takes into account the system dynamics and ensures that all states are in $Z$. A state in $\cpi(Z)$ is robust in the sense that all future possible next states $f(x,u)$ are contained in $Z$ \textit{despite} uncertainty about which $u \in \indep_\ctrl (x)$ is chosen.

Operator $\sipre_\inv$ below identifies states that can stay in $Z$ and remain safely in $\inv$ without coordination
\begin{equation}
\sipre_\inv(Z) = Z \cap \cpi(Z) \cap \inv. \label{eqnsipre}
\end{equation}
By iterating this operator $k$ times, we can identify a region of the state space that remains in $\inv$ for $k$ time steps despite communication losses. Both operators are simple modifications on standard controllable predecessor operators \cite{tabuada2009verification}. 

\subsection{Remarks about Algorithmic Implementation}

Set intersection, union, negation, and projection are the main operations that are required to compute \Cref{eqnuncoordsafe} and \Cref{eqnsipre} exactly.
In a continuous domain, support for these algebraic operations may only be possible to encode for a specific set of system dynamics and constraints (consider for instance linear system dynamics and constraints given as unions of polyhedra). However in the scenario where state and inputs spaces are finite, binary decision diagrams (BDDs)\cite{bryant1992symbolic} are an efficient data structure that supports all of the aforementioned operations. Instead of imposing constraints on the system dynamics and safety region, we opt for the finite case by using a grid to approximate a continuous domain.
Moreover, there exists a rich theoretical literature of abstraction methods \cite{tabuada2009verification} \cite{reissig2017feedback} and accompanying software tools such as \cite{rungger2016scots} which construct approximately similar finite systems such that \Cref{assumfinite} is satisfied, even if the state and input spaces of system $f$ are dense, continuous subsets of Euclidean space.
\begin{assumption} \label{assumfinite}
Both $\state$ and $\control$ are finite sets.
\end{assumption}

%% file: intermittent.tex
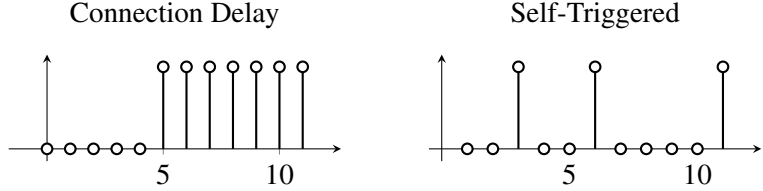
\begin{figure}
\centering
\begin{tikzpicture}
\begin{axis}[%
                standard,
                domain = 0:15,
                samples = 16,
                xlabel={},
                title={Connection Delay},
                yticklabels={,,},
                height = 3cm,
                width = 6cm,
                ymajorticks=false]
				                
                \addplot+[ycomb, black, thick] plot coordinates { (0,0) (1,0) (2,0) (3,0) (4,0) (5,1) (6,1) (7,1) (8,1) (9,1) (10,1) (11,1)};
\end{axis}
\end{tikzpicture}
~~~~~~~~~~
\begin{tikzpicture}
\begin{axis}[%
                standard,
                domain = 0:15,
                samples = 16,
                xlabel={},
                title ={Self-Triggered},
                yticklabels={,,},
                height = 3cm,
                width = 6cm,
                ymajorticks=false]
				                
                \addplot+[ycomb, black, thick] plot coordinates { (1,0) (2,0) (3,1) (4,0) (5,0) (6,1) (7,0) (8,0) (9,0) (10,0) (11,1)};
\end{axis}

\end{tikzpicture}

\caption{Two scenarios with intermittent connections. A high value signifies an established connection.} \label{figconnections}
\end{figure}

\section{Applications}

We consider two applications. One is to characterize latency requirements for a wireless communication system and the other is a design for a self-triggered coordination scheme.

\subsection{Maximum Allowed Connection Delay} \label{ssmaxdelay}

Our first application involves $N$ agents that seek to establish a wireless communication channel subject to a maximum connection delay  $D \in \mathbb{N}$. Once a connection is established, it is assumed to be maintained as in the left of \Cref{figconnections} where $D = 5$. If all agents attempt to initiate a connection starting at time $k$, then they are able to jointly choose a control input starting at time $k + D$.
\begin{definition}  
A system in state $x[k]$ at time $k$ is robustly safe to connection initialization delays of length $D$ if $x[k,\infty ) \in \inv$ for all trajectories $x[k,\infty )$ generated by the time varying closed loop system 
\begin{align}
x[k+1] &\in  f \circ \indep_\ctrl(x[k]) \text{ if } k \in [k,k+D) \\
x[k+1] &\in  f \circ \ctrl(x[k]) \text{ if } k \in [k+D, \infty)
\end{align} 
where we adopt the convention $[k, k+D) = \emptyset$ if $D = 0$. 
\end{definition}

The approach to generating the set of states that are robust to connection initialization delays of length $D$ is as follows. We first identify an invariance set $\safe$ where the system $f \circ \ctrl$ remains in $\inv$ along an infinite horizon $[k+D,\infty)$ once $x[k+D] \in \safe$. Invariance set $\safe$ is distinct from safe set $\inv$ because a state $x[k] \in \inv \setminus \safe$ satisfies the safety condition at time $k$ but is not guaranteed to do so along an infinite horizon. With set $\safe$, we then iterate $\sipre_\inv(\safe)$ $D$ times to identify the states that are guaranteed to reach $\safe$ at time $k+D$ without exiting $\inv$ within $[k, k+D)$. 

To identify $\safe$, we define operators that are analogous to $\cpi$ and $\sipre$, except that $\indep_\ctrl$ is replaced with $\ctrl$ 
\begin{align}
\pre(Z)=& \{x: x \in \proj_\state(\ctrl)\} \cap \: \{ x: \emptyset \neq f(x,u) \subseteq Z \text{ for all } u \in C(x)\} \\
\spre_\inv(Z) =& Z \cap \pre(Z) \cap \inv
\end{align}

\begin{lemma}
Let $\safe := \lim_{i \rightarrow \infty}  \spre_\inv^i(\state)$. Then all trajectories $x[k + D, \infty)$ such that $x[k+D] \in \safe$ will never intersect the unsafe set $\inv^C$.
\end{lemma} 
\begin{proof}
The Tarski fixed point theorem \cite{tarski1955lattice} ensures that the limit on the right hand side exists and is unique if $\state$ is a finite set and $\spre_\inv$ is a monotone operator. \Cref{assumfinite} ensures that $\state$ is finite, and monotonicity of $\spre_\inv$ with respect to the set containment ordering can easily be verified. Note that $\inv = \spre_\inv^1(\state)$. Membership of state $x[k]$ in set $\spre\inv^{i+1}(\state)$ ensures that both $x[k], x[k+1] \in \safe$. By induction, given $x[k+D] \in \spre_\inv^i(\state)$ and $i > 0$, trajectories from system $f \circ \ctrl$ will remain in $\inv$ along the interval $[k + D, k + D + i)$ . Because the limit set exists, $\lim_{i \rightarrow \infty}  \spre_\inv^i(\state)$ is the set of points that are safe along the interval $[k + D, \infty)$. 
\end{proof}

Building on the previous lemma, iterating $\sipre$ $D$ times yields a region where all trajectories of length $D$ are safe without coordination.
The closed loop system under $\indep_\ctrl$ must never exit $\inv$ within the interval $[k, k+D)$, and also must terminate at $x[k+D] \in \safe$ so that the system under $\ctrl$ can ensure safety along the infinite horizon $[k+D,\infty)$. 

\begin{proposition} \label{propcoordfp}
Let $\safe := \lim_{i \rightarrow \infty} \scp_{\inv}^i(\safe)$. Then $\sipre_\inv^k (\safe)$ is the set of states that are safe under $\indep_\ctrl$ for $k-1$ time steps.
\end{proposition}
\begin{proof}
Suppose $x[0] \in \sipre_\safe^k (\safe)$. The set of possible states for $x[1]$ under controller $\indep_\ctrl$ is uniquely defined as $\sipre_\safe^{k-1} (\safe)$ and is non-empty. By induction, a sequence $x\sig = x[0]\ldots x[k]$ generated by closed loop system $f \circ \indep_\ctrl$ must satisfy $x[j] \in \sipre_\safe^{k-j} (\safe)$ for all $j \in [0,k]$. By definition $\sipre_\safe^{0} (\safe) = \safe$. 
\end{proof}

\subsection{Self-triggered coordination}

It is also possible to design a scheduler for triggering communication amongst agents. Each agent maintains a countdown for the latest time communications can be initiated. As the system executes, this time is updated to provide a constantly changing upper bound on the latest time the agents need to communicate. For clarity, we assume that the connection initialization delay as in the previous section is $D = 0$.

The fixed point computation in \Cref{propcoordfp} yields a sequence of disjoint sets. Define $T: [0, F] \maps \pow{\state}$ such that
\begin{align}
T(k) = 
\left\{ 
\begin{array}{ll}
\sipre_\inv^k (\safe) \setminus \sipre_\inv^{k+1}(\safe) & \text{ if } k < F\\
\sipre_\inv^k (\safe) & \text{ if } k = F
\end{array} \right.
\end{align} 
where $F \in \NN$ is the first value where the sequence reaches a fixed point
\begin{align}
F = \text{argmin}_{i \in \NN_{\geq 0}} \sipre_{\inv}^{i+1}(\safe) = \sipre_{\inv}^{i}(\safe). 
\end{align}

A modified inverse function $\hat{T}^{-1}: \state \maps [0,F]$ is given by:
\begin{align}
\hat{T}^{-1}(x) = \{i \in [1,F]: x \in T(i) \}. \label{eqtriggercountdown}
\end{align} 
Because the collection $T(1), \ldots, T(F)$ consists of disjoint sets, $\hat{T}^{-1}(x)$ is well defined (i.e. a singleton set) for each $x \in \safe$. 
Because each agent has access to $\hat{T}$ and the state $x$, they can independently determine the unique value for $i$ such that $x \in T(i)$. A countdown with initial value $i$ is then initialized for each agent. When that value reaches $i=0$ then the agents coordinate by selecting an action and also initialize a new countdown timer. 
This framework exhibits reduced communication overhead compared to a centralized architecture, while also preserving the guarantees that are otherwise impossible with a fully decentralized and coordination free controller architecture. 

The self-triggered system is defined by augmenting the original system with a countdown that resets after coordination has been triggered. 

\begin{definition} \label{defselftriggered}
The system with a self-triggered communication architecture satisfies the following dynamics. 
\begin{align}
x[k+1] &= \left\{\begin{array}{ll}
f \circ \indep_\ctrl(x[k]) & \text{ if } i[k] > 0\\
f \circ \ctrl(x[k]) & \text{ if } i[k] = 0
\end{array}\right. \label{eqnselftrigstatetrans}\\
i[k+1] &= \left\{ \begin{array}{ll}
i[k] - 1 & \text{ if } i[k] > 0 \\
\hat{T}^{-1}(x[k+1]) & \text{ if } i[k] = 0 \\
\end{array}\right. 
\end{align} 
Note that when $i[k] = 0$, the counter is reset to $\hat{T}^{-1}(x[k+1])$ after the state transition from \Cref{eqnselftrigstatetrans} occurs.
\end{definition}

\begin{proposition}
If $x[k] \in \safe$, then all trajectories $x[k,\infty)$ under the self-triggered communication system from \Cref{defselftriggered} will remain inside $\inv$. 
\end{proposition}

%% file: examples.tex
\section{Examples}

In each of our examples, we use a modified version of the \text{SCOTS} symbolic controller synthesis toolbox \cite{rungger2016scots}, which takes a continuous control system and creates a finite state machine that serves as an abstract representation over which a controller is synthesized. In addition to modifications to compute \Cref{defmeshindependentC} and \Cref{eqnsipre}, we exploit internal system dependencies to reduce the computation time of the abstraction \cite{gruber2017sparsity}.
Creating the discrete abstraction depends on parameters such as the grid size and granularity.
Consider a set $\PP = \prod_{i=1}^N \PP_i$ and a discretization parameter $\disc \in \mathbb{R}_{> 0}^N$. Its corresponding discretization grid is $[\PP]_\disc := \prod_{i=1}^N [\PP_i]_{\disc_i}$  where $[\PP_i]_{\disc_i} := \{a \in \PP_i: a = k \disc_i \text{ with } k \in \mathbb{Z} \}$ is a grid over a single dimension. A full introduction to the underlying theory appears in \cite{tabuada2009verification} and is beyond the scope of this paper.

\subsection{Invariance in a Circle} \label{subseccircle}

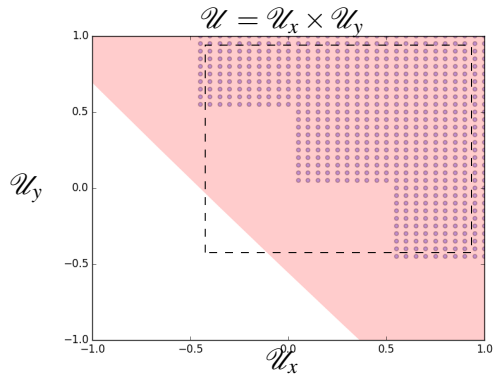
\begin{figure}
\centering
\begin{tikzpicture}
\begin{scope}
\node at (0,0) {\includegraphics[width = .42\columnwidth]{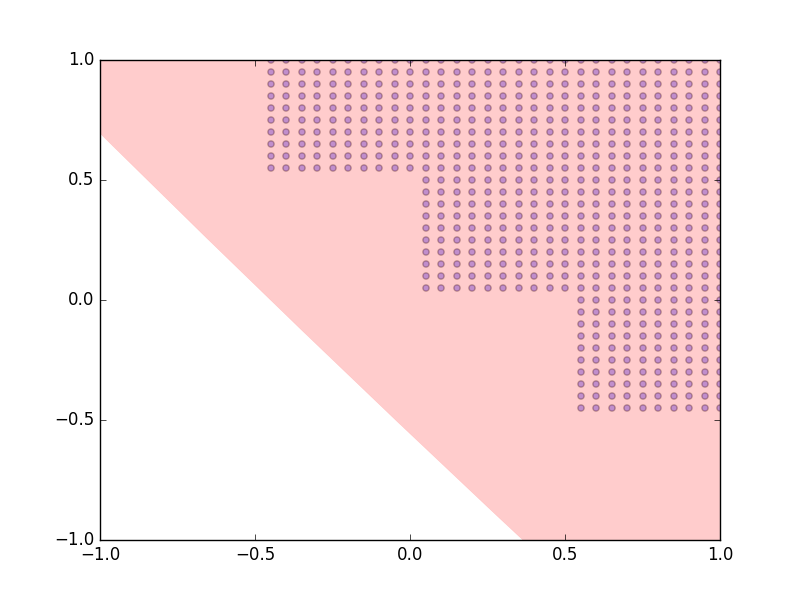}};
\node at (0,2.2) {{$\control = \control_x \times \control_y$}};
\node at (0,-2.3) {$\control_x$};
\node at (-3.4,0) {$\control_y$};
\draw[dashed] (-1.02, -.86) rectangle (2.52,1.9);
\end{scope} 
\end{tikzpicture}
\caption{Individual dots represent the synthesized safe control set from SCOTS under $\ctrl(x)$ at point $x= (x_1, x_2) = (-.62, -.5)$. Without discretization, the true safe action space would be the shaded region in red. The dashed box shows the possible coordination-free actions $\indep_\ctrl(x)$, which is not contained in the safe action space. Importantly, the synthesized safe inputs are a subset of the true set. Note that $||x||_2 \approx .796$, which is near the boundary of $\inv$.} \label{figsafecontrols}
\end{figure}

Two agents each have control over different axes and both need to remain within a circular region. 
\begin{align}
\begin{array}{ll}
\dot{x}_1 &= u_1 \\
\dot{x}_2 &= u_2
\end{array} \label{eqn2ddynamics}
\end{align}

Let $\state = \control = [-1,1] \times [-1,1]$. Although the dynamics are independent, the safety region is a circle with a radius $0.8$ so $\inv = \{(x_1,x_2): x_1^2 + x_2^2 \leq .64\}$ so both agents must coordinate with one another to avoid exiting $\inv$ near the boundary. 
It is clear that the system can always enforce safety within $\inv$ simply by picking a control input $(u_1, u_2) := -(x_1,x_2)$. 

A discretization of the system dynamics is constructed with a sampling period of $t = .01$. The state space grid $[\state]_\disc$ is constructed with $\disc = [.01,.01]$ and input space grid is $[\control]_\epsilon$ with $\epsilon = [.05, .05]$. 
\Cref{figsafecontrols} depicts all safe control inputs at $(x_1, x_2) = (-.62, .5)$ which is near the boundary of $\inv$. 
The staircase shape of the boundary between the safe and unsafe inputs is due to the discretization of the dynamics. 
Inputs towards the upper right move the state to the interior of $\inv$, while safe inputs at the lower left hug the boundary between $\inv$ and $\inv^C$. 
If both systems jointly pick low values for $u_1$ and $u_2$ then a violation occurs, however both agents can pick $u_1,u_2 = -1$ if the other agent concedes and chooses a higher value. 

\begin{figure}
\centering
\includegraphics[width = .49\columnwidth]{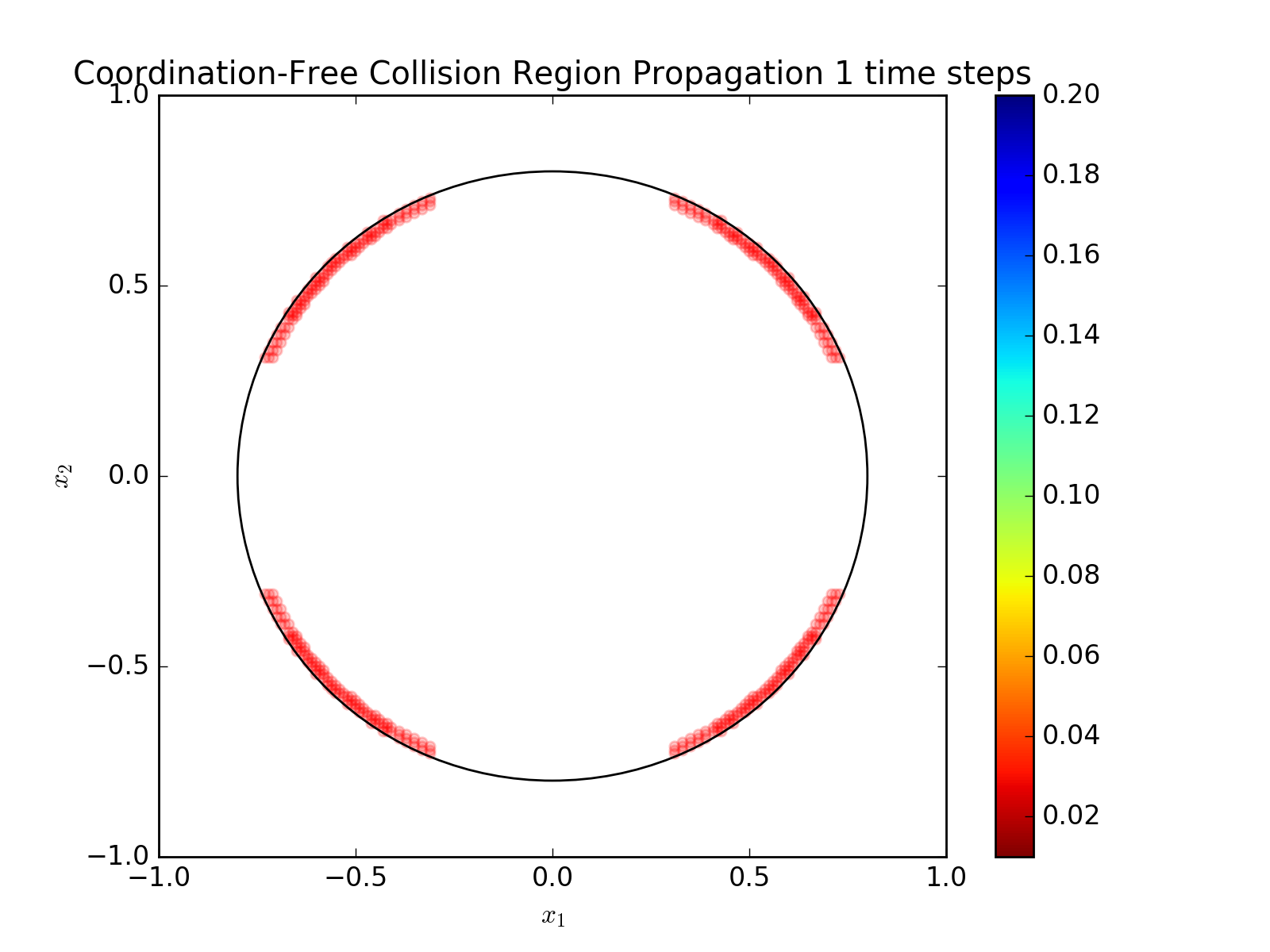}
\includegraphics[width = .49\columnwidth]{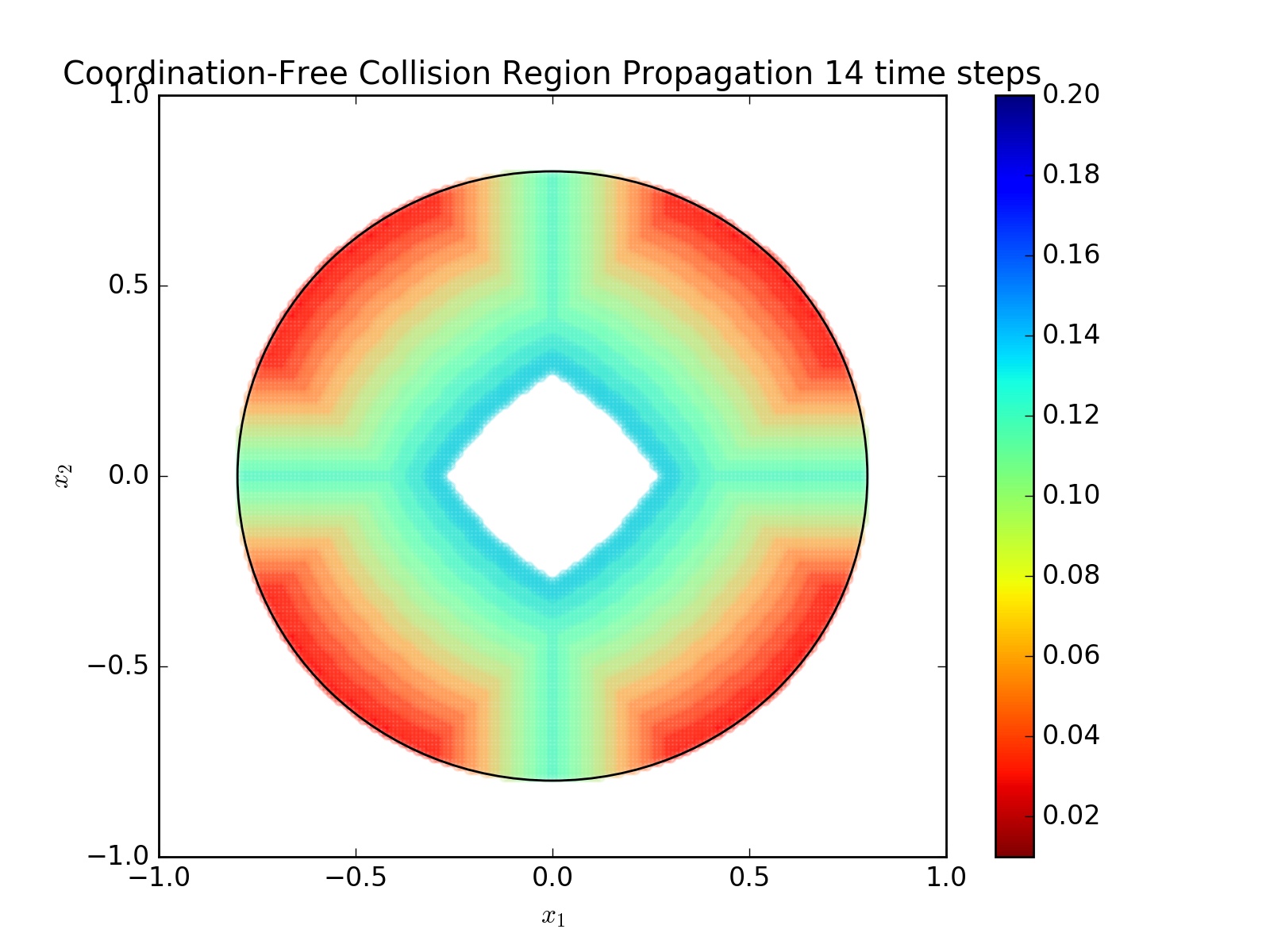}
\includegraphics[width = .49\columnwidth]{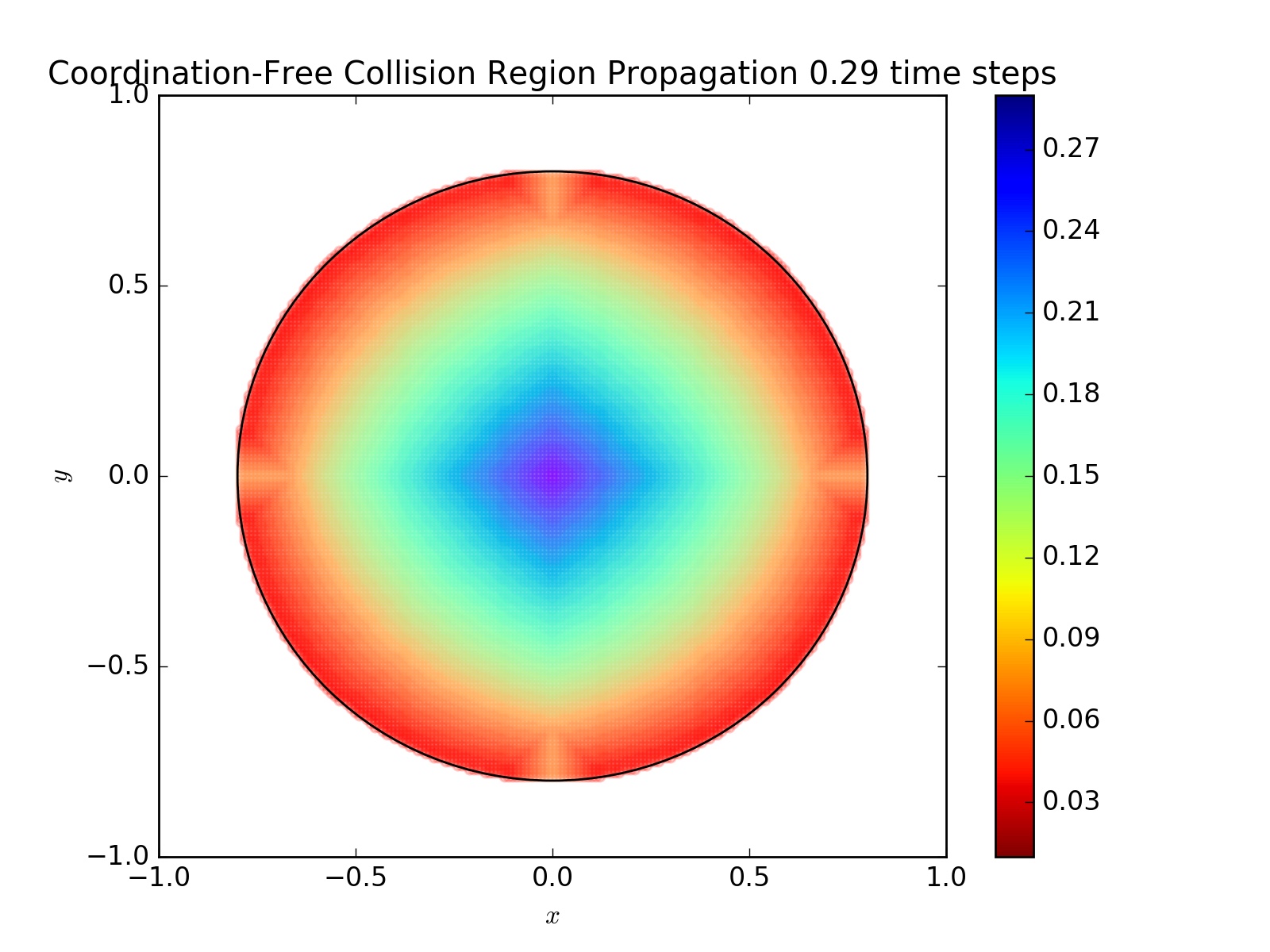}
\caption{Multiple snapshots at $i = 1, 14, 29$ as the region $\safe \setminus \scp_{\inv}^i(\safe)$ grows. One can alternatively visualize $\scp_{\inv}^i(\safe)$ as a shrinking interior white region as the length of the communication-free interval grows. Red regions represent areas where the system will imminently exit $\inv$ unless the two agents coordinate their actions, while blue regions in the interior are only unsafe if the agents do not coordinate for a prolonged period. A fixed point was reached at $i=29$.} \label{figcirclecoord} 
\end{figure} 
\Cref{figcirclecoord} depicts the propagation at various time steps of the coordination-free region via the $\sipre$ operator in \Cref{seccoordregion}.
\Cref{figcirclecoord} shows that a system beginning at the origin can experience an uncoordinated collision is possible after $29$ discrete time steps which under sampling period $t = .29$ corresponds to an interval of length .29 in continuous time. 
However for the continuous system the worst case time step is roughly twice as much $.8/\sqrt{2} \approx .565$, which is the case when $u_1, u_2 \in \{-1, 1\}$ and maintain constant values over time.
This is mainly due to the discretization errors that arise when abstracting the continuous system to a discrete one.
Note that the discretization error does not jeopardize the safety guarantee. Rather, the discrete case underestimates how much time is available for agents to avoid communication, thus providing a more conservative guarantee. 

\subsection{Intersection Collision Avoidance}
Consider two vehicles that are approaching an intersection with no stop sign or a traffic signal. They are controlled independently but each are equipped with V2V radios and may communicate with one another. They also are equipped with enough sensors to identify the position and velocity of all vehicles near the intersection. We consider a simple set of system dynamics given by
\begin{align}
\dot{p}_i &= v_i\\
\dot{v}_i &= u_i - Kv^2 \label{eqnfwdvel}
\end{align}
with some constant $K = .2$. 
A higher value for $k$ signifies higher air drag.
Let $\mathcal{P}_1, \mathcal{P}_2 = [-10,10]$ and $\mathcal{V}_1, \mathcal{V}_2 = [0,3]$. The state space is $\state := \prod_{i=1}^2 \left( \mathcal{P}_i \times \mathcal{V}_i \right)$ and $\control := \prod_{i=1}^2 [-1,1]$. The invariant region is the region where at least one vehicle is outside the intersection and no collision has occurred and is succinctly encoded as the set
\begin{align}
\inv := \{x: (|p_1| \geq 2) \vee (|p_2| \geq 2)\}. 
\end{align}

We use the SCOTS toolbox to synthesize a supervisory controller $C$ and compute its corresponding invariance  region $\safe$ with the procedure in \Cref{ssmaxdelay}. The system dynamics discretization used a sampling period of $t = .2$, state space grid $[\state]_\disc$ parameter $\disc = [.1,.1,.1,.1]$ and input space grid $[\control]_\epsilon$ parameter $\epsilon = [.1, .1]$. 

After synthesizing controller $C$, its  decomposed counterpart $\indep_\ctrl$ is analyzed. 
Within $\safe^C$ even a centralized controller is unable to guarantee that a collision will \textit{not} occur. This unsafe region is to be avoided and communication is necessary to avoid it. \Cref{figinterunsafe} depicts the 3D projection of $\safe^C$ and the evolution of the unsafe region $(\sipre_\inv^D (\safe))^C$ with no communication. 

\begin{figure}  \label{figinterunsafe}
\centering
\includegraphics[width = .5 \columnwidth]{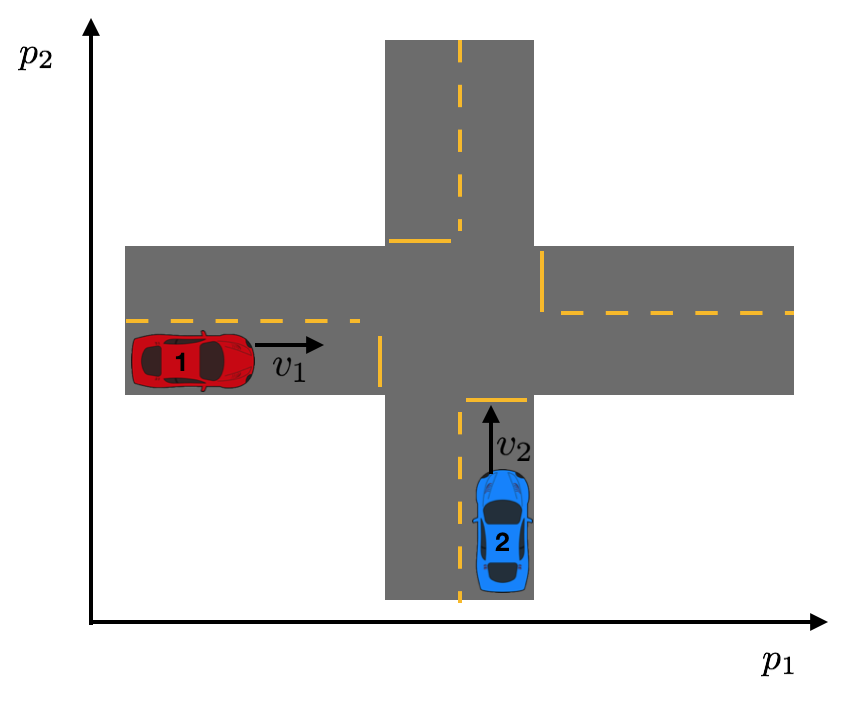}
\caption{Intersection Collision Avoidance}
\end{figure}
\begin{figure}
\includegraphics[width = .5\columnwidth]{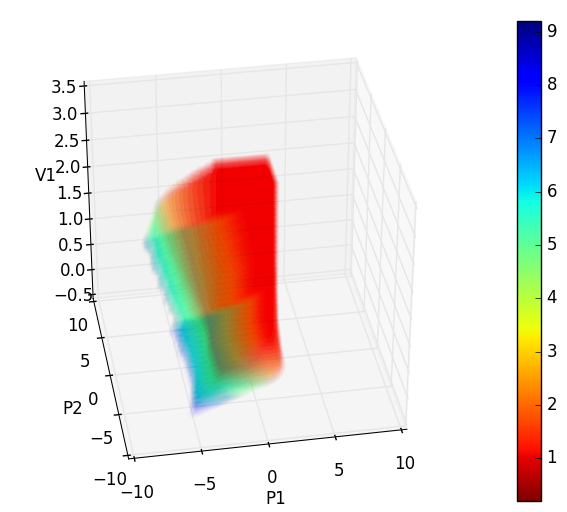}
\includegraphics[width = .5\columnwidth]{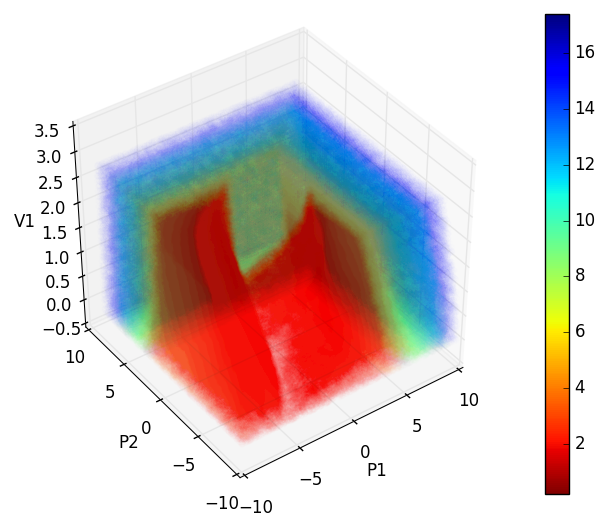}
\caption{(Left) Three dimensional projection of the four dimensional unsafe region $\safe^C$ for centralized controller $\ctrl$ with $v_2 = 2.8$ held constant. Color scale shows the earliest potential collision time. (Right) Figure shows the unsafe action region $(\sipre_\inv^D (\safe))^C$ for the system $f \circ \indep_\ctrl(x)$ expand as communication delay $D$ increases.}
\end{figure}

\subsection{Self-Triggered Coordination in a 2D Gridworld}

Let there be $N = 2$ agents navigating a 2D grid. Both agents have identical dynamics to \Cref{eqn2ddynamics} as shown below with superscripts $i = 1,2$ as indexes for each agent.
\begin{align}
\begin{array}{ll}
\dot{x}_1^i &= u_1^i \\
\dot{x}_2^i &= u_2^i
\end{array} \label{eqn2ddynamics}
\end{align}
The sets $\state^i = [-.2,.2] \times [-.2, .2]$ and $\control^i = [-1,1] \times [-1,1]$ for both $i=1,2$. 
A collision has occurred between both agents in the region
\begin{align}
\inv^C = \{(x^1,x^2) \in \state^1 \times \state^2: \max(|x_1^1 - x_1^2|, |x_2^1 - x_2^2|) < 0.1 \}.
\end{align}

SCOTS is again used to synthesize a centralized controller for the system. The discrete abstraction was constructed with sampling period $\tau = .01$, state space grid $[\state]_\eta$ with parameter $\eta=[.01, .01, .01, .01]$, and input space grid $[\control]_\epsilon$ with parameter $\epsilon = [.2,.2,.2,.2]$.
\Cref{figselftriggered} shows the trajectory of the system with the self-triggering implementation and how $\hat{T}^{-1}(x[k])$ as defined in \Cref{eqtriggercountdown} varies with respect to time. 
 
\begin{figure}
\centering
\includegraphics[height = .34\columnwidth]{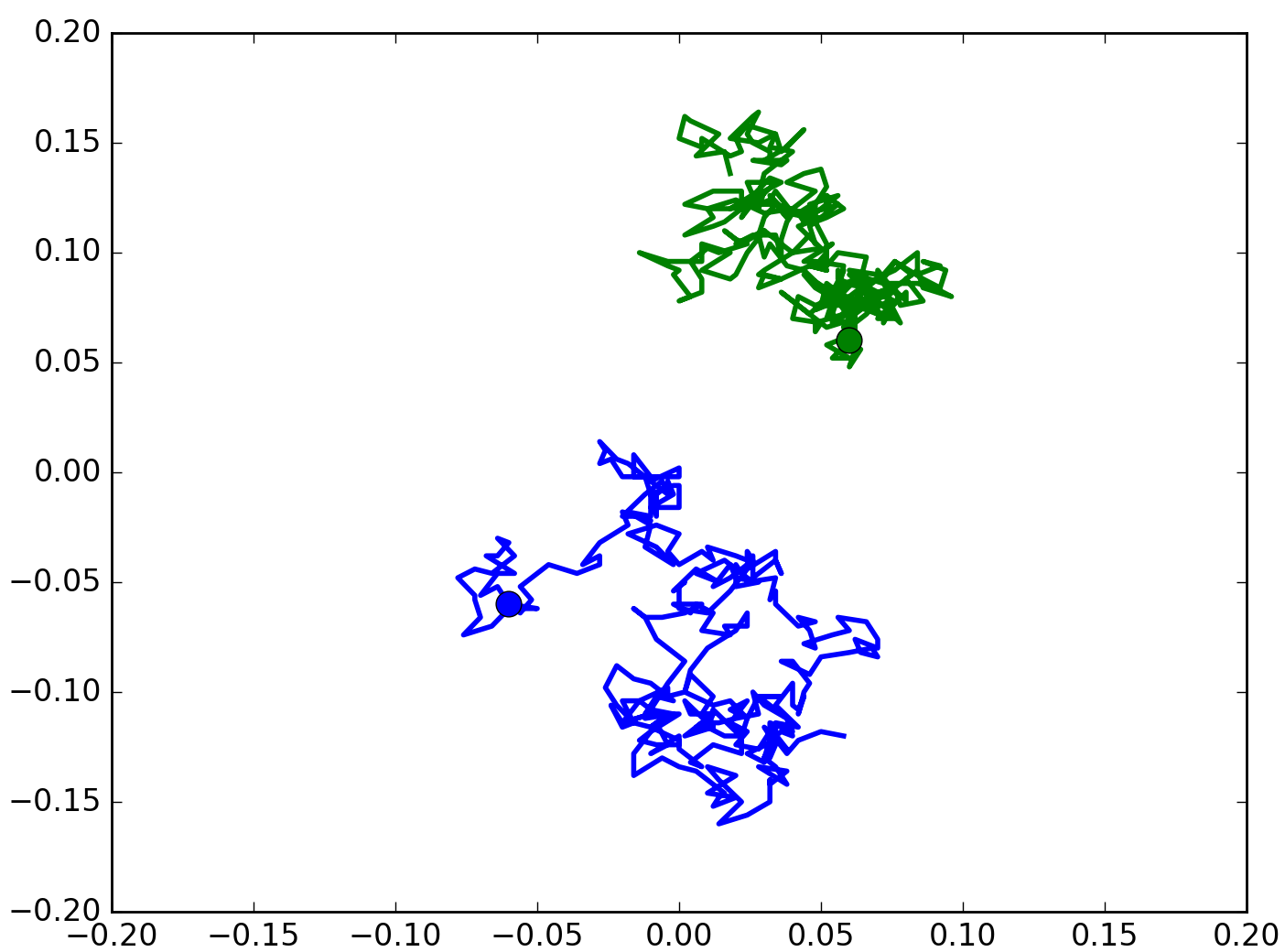}
\includegraphics[height = .34\columnwidth]{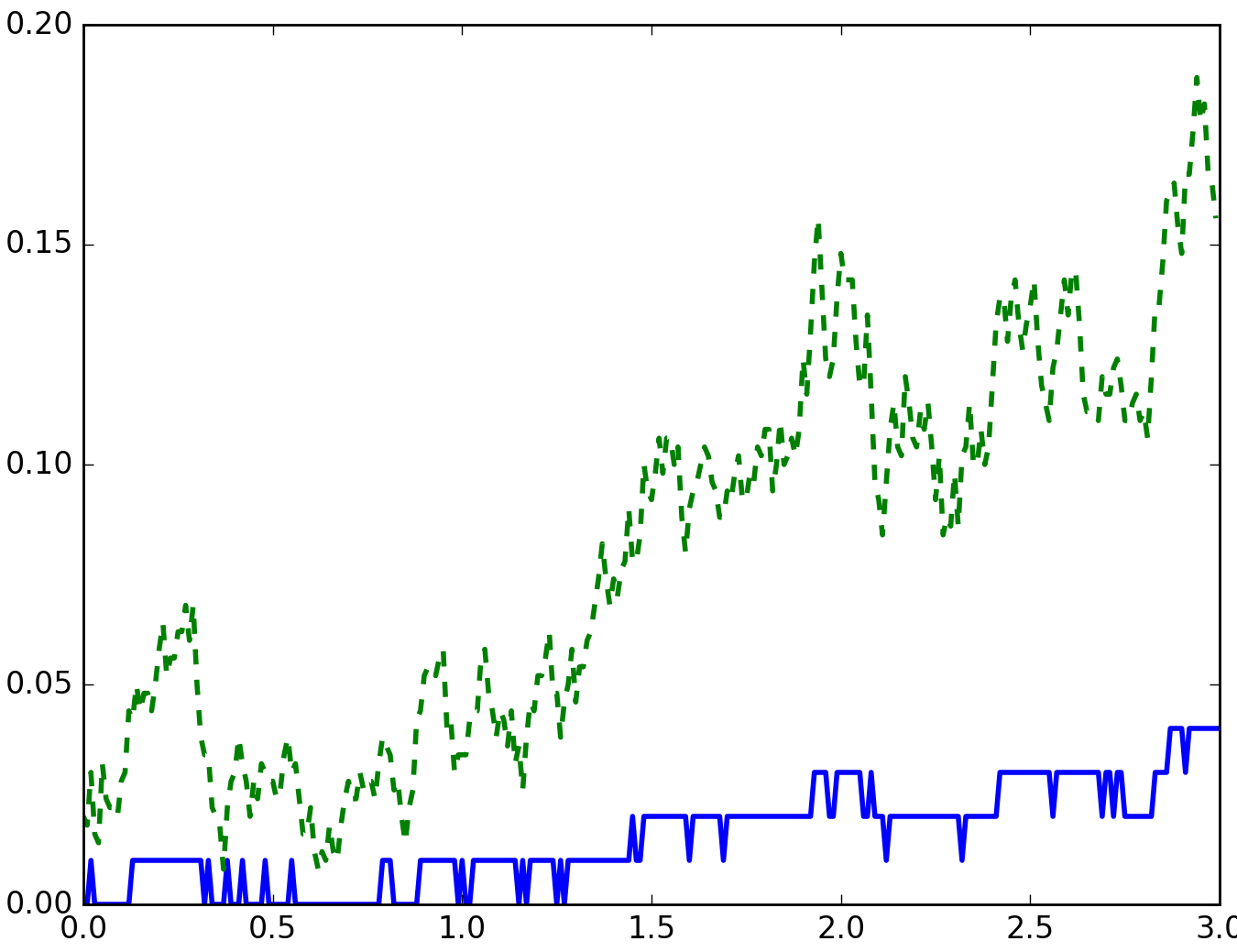}
\caption{ (Left) Trajectories of both systems (Right) The solid line is the value for $\hat{T}^{-1}(x[k])$ which underapproximates the actual time to when a collision is inevitable $\safe^C$. Because \Cref{eqn2ddynamics} is fully actuated, the safe set $\safe$ and the invariance region $\inv$ are identical.} \label{figselftriggered}
\end{figure}

%% file: conclusion.tex
\section{Conclusion}

We have presented a method to analyze when communication is necessary in order for a distributed controller to satisfy a safety requirement. While the current implementation deals with memoryless controllers future work will look into control policies with memory, time varying connectivity, and an application to richer specifications including those expressible in temporal logic.

%% file: main.bbl
\begin{thebibliography}{10}
\providecommand{\bibitemdeclare}[2]{}
\providecommand{\surnamestart}{}
\providecommand{\surnameend}{}
\providecommand{\urlprefix}{Available at }
\providecommand{\url}[1]{\texttt{#1}}
\providecommand{\href}[2]{\texttt{#2}}
\providecommand{\urlalt}[2]{\href{#1}{#2}}
\providecommand{\doi}[1]{doi:\urlalt{http://dx.doi.org/#1}{#1}}
\providecommand{\bibinfo}[2]{#2}

\bibitemdeclare{manual}{rightofway}
\bibitem{rightofway}
\bibinfo{author}{Federal~Aviation \surnamestart Administration\surnameend}
  (\bibinfo{year}{2017}): \emph{\bibinfo{title}{Code of Federal Regulations,
  Title 14}}.

\bibitemdeclare{article}{brunner2015communication}
\bibitem{brunner2015communication}
\bibinfo{author}{Florian~D \surnamestart Brunner\surnameend},
  \bibinfo{author}{TMP \surnamestart Gommans\surnameend}, \bibinfo{author}{WPMH
  \surnamestart Heemels\surnameend} \& \bibinfo{author}{Frank \surnamestart
  Allg{\"o}wer\surnameend} (\bibinfo{year}{2015}):
  \emph{\bibinfo{title}{Communication Scheduling in Robust Self-Triggered MPC
  for Linear Discrete-Time Systems}}.
\newblock {\sl \bibinfo{journal}{IFAC-PapersOnLine}}
  \bibinfo{volume}{48}(\bibinfo{number}{22}), pp. \bibinfo{pages}{132--137},
  \doi{10.1016/j.ifacol.2015.10.319}.

\bibitemdeclare{article}{bryant1992symbolic}
\bibitem{bryant1992symbolic}
\bibinfo{author}{Randal~E \surnamestart Bryant\surnameend}
  (\bibinfo{year}{1992}): \emph{\bibinfo{title}{Symbolic Boolean manipulation
  with ordered binary-decision diagrams}}.
\newblock {\sl \bibinfo{journal}{ACM Computing Surveys (CSUR)}}
  \bibinfo{volume}{24}(\bibinfo{number}{3}), pp. \bibinfo{pages}{293--318},
  \doi{10.1145/136035.136043}.

\bibitemdeclare{article}{chen2016general}
\bibitem{chen2016general}
\bibinfo{author}{Mo~\surnamestart Chen\surnameend}, \bibinfo{author}{Sylvia~L
  \surnamestart Herbert\surnameend}, \bibinfo{author}{Mahesh~S \surnamestart
  Vashishtha\surnameend}, \bibinfo{author}{Somil \surnamestart
  Bansal\surnameend} \& \bibinfo{author}{Claire~J \surnamestart
  Tomlin\surnameend} (\bibinfo{year}{2018}): \emph{\bibinfo{title}{A general
  system decomposition method for computing reachable sets and tubes}}.
\newblock {\sl \bibinfo{journal}{IEEE Transactions on Automatic Control}},
  \doi{10.1109/TAC.2018.2797194}.

\bibitemdeclare{inproceedings}{dallal16}
\bibitem{dallal16}
\bibinfo{author}{Eric \surnamestart Dallal\surnameend} \&
  \bibinfo{author}{Paulo \surnamestart Tabuada\surnameend}:
  \emph{\bibinfo{title}{Decomposing Controller Synthesis for Safety
  Specifications}}.
\newblock In: {\sl \bibinfo{booktitle}{CDC2016}},
  \doi{10.1109/CDC.2016.7799148}.

\bibitemdeclare{article}{GOMMANS201559}
\bibitem{GOMMANS201559}
\bibinfo{author}{T.M.P. \surnamestart Gommans\surnameend} \&
  \bibinfo{author}{W.P.M.H. \surnamestart Heemels\surnameend}
  (\bibinfo{year}{2015}): \emph{\bibinfo{title}{Resource-aware MPC for
  constrained nonlinear systems: A self-triggered control approach}}.
\newblock {\sl \bibinfo{journal}{Systems and Control Letters}}
  \bibinfo{volume}{79}, pp. \bibinfo{pages}{59 -- 67},
  \doi{10.1016/j.sysconle.2015.03.003}.
\newblock
  \urlprefix\url{http://www.sciencedirect.com/science/article/pii/S0167691115000481}.

\bibitemdeclare{inproceedings}{gruber2017sparsity}
\bibitem{gruber2017sparsity}
\bibinfo{author}{Felix \surnamestart Gruber\surnameend},
  \bibinfo{author}{Eric~S \surnamestart Kim\surnameend} \&
  \bibinfo{author}{Murat \surnamestart Arcak\surnameend}
  (\bibinfo{year}{2017}): \emph{\bibinfo{title}{Sparsity-Aware Finite
  Abstraction}}.
\newblock In: {\sl \bibinfo{booktitle}{CDC2017}},
  \doi{10.1109/CDC.2017.8263995}.

\bibitemdeclare{inproceedings}{heemels2012introduction}
\bibitem{heemels2012introduction}
\bibinfo{author}{WPMH \surnamestart Heemels\surnameend},
  \bibinfo{author}{Karl~Henrik \surnamestart Johansson\surnameend} \&
  \bibinfo{author}{Paulo \surnamestart Tabuada\surnameend}
  (\bibinfo{year}{2012}): \emph{\bibinfo{title}{An introduction to
  event-triggered and self-triggered control}}.
\newblock In: {\sl \bibinfo{booktitle}{Decision and Control (CDC), 2012 IEEE
  51st Annual Conference on}}, \bibinfo{organization}{IEEE}, pp.
  \bibinfo{pages}{3270--3285}, \doi{10.1109/CDC.2012.6425820}.

\bibitemdeclare{inproceedings}{kim2015compositional}
\bibitem{kim2015compositional}
\bibinfo{author}{Eric~S \surnamestart Kim\surnameend}, \bibinfo{author}{Murat
  \surnamestart Arcak\surnameend} \& \bibinfo{author}{Sanjit~A \surnamestart
  Seshia\surnameend} (\bibinfo{year}{2015}):
  \emph{\bibinfo{title}{Compositional controller synthesis for vehicular
  traffic networks}}.
\newblock In: {\sl \bibinfo{booktitle}{Decision and Control (CDC), 2015 IEEE
  54th Annual Conference on}}, \bibinfo{organization}{IEEE}, pp.
  \bibinfo{pages}{6165--6171}, \doi{10.1109/CDC.2015.7403189}.

\bibitemdeclare{article}{meyer2017compositional}
\bibitem{meyer2017compositional}
\bibinfo{author}{Pierre-Jean \surnamestart Meyer\surnameend},
  \bibinfo{author}{Antoine \surnamestart Girard\surnameend} \&
  \bibinfo{author}{Emmanuel \surnamestart Witrant\surnameend}
  (\bibinfo{year}{2017}): \emph{\bibinfo{title}{Compositional abstraction and
  safety synthesis using overlapping symbolic models}}.
\newblock {\sl \bibinfo{journal}{IEEE Transactions on Automatic Control}},
  \doi{10.1109/TAC.2017.2753039}.

\bibitemdeclare{article}{nuzzo2015platform}
\bibitem{nuzzo2015platform}
\bibinfo{author}{Pierluigi \surnamestart Nuzzo\surnameend},
  \bibinfo{author}{Alberto~L \surnamestart Sangiovanni-Vincentelli\surnameend},
  \bibinfo{author}{Davide \surnamestart Bresolin\surnameend},
  \bibinfo{author}{Luca \surnamestart Geretti\surnameend} \&
  \bibinfo{author}{Tiziano \surnamestart Villa\surnameend}
  (\bibinfo{year}{2015}): \emph{\bibinfo{title}{A platform-based design
  methodology with contracts and related tools for the design of cyber-physical
  systems}}.
\newblock {\sl \bibinfo{journal}{Proceedings of the IEEE}}
  \bibinfo{volume}{103}(\bibinfo{number}{11}), pp. \bibinfo{pages}{2104--2132},
  \doi{10.1109/JPROC.2015.2453253}.

\bibitemdeclare{article}{reissig2017feedback}
\bibitem{reissig2017feedback}
\bibinfo{author}{Gunther \surnamestart Reissig\surnameend},
  \bibinfo{author}{Alexander \surnamestart Weber\surnameend} \&
  \bibinfo{author}{Matthias \surnamestart Rungger\surnameend}
  (\bibinfo{year}{2017}): \emph{\bibinfo{title}{Feedback refinement relations
  for the synthesis of symbolic controllers}}.
\newblock {\sl \bibinfo{journal}{IEEE Transactions on Automatic Control}}
  \bibinfo{volume}{62}(\bibinfo{number}{4}), pp. \bibinfo{pages}{1781--1796},
  \doi{10.1109/TAC.2016.2593947}.

\bibitemdeclare{article}{rungger2016compositional}
\bibitem{rungger2016compositional}
\bibinfo{author}{Matthias \surnamestart Rungger\surnameend} \&
  \bibinfo{author}{Majid \surnamestart Zamani\surnameend}
  (\bibinfo{year}{2016}): \emph{\bibinfo{title}{Compositional Construction of
  Approximate Abstractions of Interconnected Control Systems}}.
\newblock {\sl \bibinfo{journal}{IEEE Transactions on Control of Network
  Systems}}, \doi{10.1109/TCNS.2016.2583063}.

\bibitemdeclare{inproceedings}{rungger2016scots}
\bibitem{rungger2016scots}
\bibinfo{author}{Matthias \surnamestart Rungger\surnameend} \&
  \bibinfo{author}{Majid \surnamestart Zamani\surnameend}
  (\bibinfo{year}{2016}): \emph{\bibinfo{title}{SCOTS: A tool for the synthesis
  of symbolic controllers}}.
\newblock In: {\sl \bibinfo{booktitle}{Proceedings of the 19th International
  Conference on Hybrid Systems: Computation and Control}},
  \bibinfo{organization}{ACM}, pp. \bibinfo{pages}{99--104},
  \doi{10.1145/2883817.2883834}.

\bibitemdeclare{inproceedings}{sadraddini2017formal}
\bibitem{sadraddini2017formal}
\bibinfo{author}{Sadra \surnamestart Sadraddini\surnameend},
  \bibinfo{author}{J{\'a}nos \surnamestart Rudan\surnameend} \&
  \bibinfo{author}{Calin \surnamestart Belta\surnameend}
  (\bibinfo{year}{2017}): \emph{\bibinfo{title}{Formal synthesis of distributed
  optimal traffic control policies}}.
\newblock In: {\sl \bibinfo{booktitle}{Proceedings of the 8th International
  Conference on Cyber-Physical Systems}}, \bibinfo{organization}{ACM}, pp.
  \bibinfo{pages}{15--24}, \doi{10.1145/3055004.3055011}.

\bibitemdeclare{book}{tabuada2009verification}
\bibitem{tabuada2009verification}
\bibinfo{author}{Paulo \surnamestart Tabuada\surnameend}
  (\bibinfo{year}{2009}): \emph{\bibinfo{title}{Verification and control of
  hybrid systems: a symbolic approach}}.
\newblock \bibinfo{publisher}{Springer Science \& Business Media},
  \doi{10.1007/978-1-4419-0224-5}.

\bibitemdeclare{article}{tarski1955lattice}
\bibitem{tarski1955lattice}
\bibinfo{author}{Alfred \surnamestart Tarski\surnameend}
  (\bibinfo{year}{1955}): \emph{\bibinfo{title}{A lattice-theoretical fixpoint
  theorem and its applications}}.
\newblock {\sl \bibinfo{journal}{Pacific journal of Mathematics}}
  \bibinfo{volume}{5}(\bibinfo{number}{2}), pp. \bibinfo{pages}{285--309},
  \doi{10.2140/pjm.1955.5.285}.

\end{thebibliography}
